\documentclass[copyright,creativecommons]{eptcs}
\providecommand{\event}{MSFP 2020}

\usepackage{amsmath}
\usepackage{amssymb}
\usepackage{amsthm}
\usepackage[all]{xy}
\usepackage{stmaryrd}

\newtheorem{thm}{Theorem}[section]
\newtheorem{lem}[thm]{Lemma}

\newtheorem{pro}[thm]{Proposition}

\theoremstyle{definition}
\newtheorem{defi}[thm]{Definition}
\newtheorem{exa}[thm]{Example}

\newcommand{\Nat}{\mathbb{N}}

\newcommand{\Real}{\mathbb{R}}
\newcommand{\Bool}{\mathbf{B}}

\newcommand{\Zero}{\mathbf{0}}
\newcommand{\Arb}{\mathbb{A}}

\newcommand{\False}{\mathbf{F}}
\newcommand{\True}{\mathbf{T}}
\newcommand{\Obser}{\mathcal{O}}
\newcommand{\basrel}{\sqsubseteq_{\Obser}}
\newcommand{\denote}[1]{\llbracket #1 \rrbracket}

\newcommand{\Trees}[1]{T #1}

\newcommand{\Power}[1]{\mathcal{P}(#1)}

\newcommand{\EfOp}{\mathsf{op}}
\newcommand{\Arr}[1]{\mathsf{ar}(#1)}
\newcommand{\leaf}[1]{\langle #1 \rangle}
\newcommand{\IEQ}{\mathcal{I}}
\newcommand{\IU}{\mathcal{B}}
\newcommand{\Ias}[2]{{#1  \sqsubseteq_{\IEQ}  #2}}
\newcommand{\Axioms}{\mathsf{A}}
\newcommand{\fimp}{\trianglelefteq}
\newcommand{\fpmi}{\trianglerighteq}
\newcommand{\Arel}{\sqsubseteq_{\alpha}}
\newcommand{\clos}[1]{[#1]}
\newcommand{\rell}[1]{\Trees{\langle #1 \rangle}}

\newcommand{\ENor}{\mathsf{nond}}
\newcommand{\EPor}{\mathsf{prob}}
\newcommand{\EsNor}{\mathsf{no}}
\newcommand{\EsPor}{\mathsf{pr}}
\newcommand{\ELo}{\mathsf{lookup}}
\newcommand{\EUp}[1]{\mathsf{update}_{#1}}
\newcommand{\EEx}[1]{\mathsf{exept}_{#1}}
\newcommand{\ECa}[1]{\mathsf{catch}_{#1}}
\newcommand{\ETi}{\mathsf{tick}}
\newcommand{\EIn}{\mathsf{input}}

\newcommand{\ClIU}{\clos{\IU}}

\newcommand{\Exc}{\mathsf{Exc}}

\title{From Equations to Distinctions:\\ Two Interpretations of Effectful Computations}
\def\titlerunning{From Equations to Distinctions: Two Interpretations of Effectful Computations}
\author{Niels Voorneveld\thanks{This research was supported by the ESF funded Estonian IT Academy research measure (project 2014-2020.4.05.19-0001).}
\email{niels.voorneveld@gmail.com}
\institute{Tallinn University of Technology}}
\def\authorrunning{Niels Voorneveld}

\begin{document}
	\maketitle
	
	\begin{abstract}
		There are several ways to define program equivalence for functional programs with algebraic effects.
		We consider two complementing ways to specify behavioural equivalence.
		One way is to specify a set of axiomatic equations, and allow proof methods to show that two programs are equivalent.
		Another way is to specify an Eilenberg-Moore algebra, which generate tests that could distinguish programs.
		These two methods are said to complement each other if any two programs can be shown to be equivalent if and only if there is no test to distinguish them.
		
		In this paper, we study a generic method to formulate from a set of axiomatic equations an Eilenberg-Moore algebra which complements it.
		We will look at an additional condition which must be satisfied for this to work. We then apply this method to a handful of examples of effects, including probability and global store, and show they coincide with the usual algebras from the literature.
		We will moreover study whether or not it is possible to specify a set of unary Boolean modalities which could function as distinction-tests complementing the equational theory.
	\end{abstract}
	
	\section{Introduction}

Program equivalence is an active field of study, allowing us to formulate when two different programs can be used interchangeably.
This can be done in two ways. One can axiomatise which programs should be considered equivalent, and derive a notion of program equivalence from those axioms. Alternatively, one can formulate theoretical tests on programs, which check whether the program satisfies a certain behavioural property. Two programs are then considered equivalent if they satisfy the same properties.

We consider such methods in the presence of functional languages with \emph{algebraic effects} in the sense of \cite{effect,alg_eff,Plotkin:2004}. Effects describe interactions a program has with the outside world. 
Because of possibly varying inputs from the outside world, the same program may produce different results at different executions.
This happens, for instance, if the program requests a random number, or reads off information from a global store location. 
Different possible continuations of a program can be combined using \emph{algebraic operators}.
E.g. we get program terms like $\EPor(P,Q)$, which probabilistically chooses fairly between executing program term $P$ and program term $Q$.

Traditionally, behaviour of algebraic effects has been formulated using \emph{algebraic equations} \cite{PlotkinPower02,Domains}.
One could for instance state that $\EPor(P,P) = P$ and $\EPor(P,Q) = \EPor(Q,P)$.
In recent research based on \cite{op_meta}, \emph{modalities} are used to formulate properties on computations that exhibit effectful behaviour \cite{modal,modal_journal,Sound}.
Using these modalities, Boolean predicates on sets of return values could be lifted to Boolean predicates on algebraic expressions over such return values. 
For instance, one might test whether the probability that a program returns an even number exceeds some threshold.
These modalities are then used as the foundation upon which a \emph{logic of program properties} is constructed, specifying a notion of \emph{behavioural equivalence} for functional languages.

In \cite{Quantitative}, such logics were generalised to \emph{quantitative logics} built using \emph{quantitative modalities}. 
In most examples of effects, it is more natural to use a singular quantitative modality, given by \emph{Eilenberg-Moore algebras} (e.g. used in \cite{HasuoGeneric}), to describe effectful behaviour.
This generalisation also enables us to describe combinations of effects more easily.

We say that an Eilenberg-Moore algebra \emph{complements} a set of equational axioms when they induce the same relation on algebraic expressions over the natural numbers. 
This then extends to them specifying the same notion of program equivalence on functional languages.
For most examples of effects, an Eilenberg-Moore algebra exists which complements the traditionally chosen set of axiomatic equations.
In this paper, we show that these algebras can be directly constructed from the algebraic relations on algebraic expressions induced by the axiomatic equations, motivating their formulation in the literature.

In general, we construct from any set of equations an Eilenberg-Moore algebra.
We show that this algebra complements the set of equations if an additional, relatively weak, property is satisfied. 
If a stronger property is satisfied, which does not hold for all examples, we can also generically generate a set of Boolean modalities complementing the axiomatic equations.

In Section~\ref{Sec:equat}, we study algebraic expressions and algebraic relations given by effects in general, and several examples in particular. In Section~\ref{Sec:EM} we look at the complementing view on effects, Eilenberg-Moore algebras, and how to construct them using algebraic relations.
Section~\ref{Sec:extra} discusses some extra topics surrounding effect descriptions, after which we look at Boolean modalities in Section~\ref{Sec:Bool}.

\section{Operations and Equations}\label{Sec:equat}

For each effect, we specify an \emph{effect signature} of algebraic operations $\Sigma$, containing effect operators $\EfOp \in \Sigma$ with an associated \emph{arity} $\Arr{\EfOp} \in \Nat \cup \{\Nat\}$.
See Subsection~\ref{sub:exa1} for examples.

\begin{defi}\label{definition:tree}
	An \emph{effect tree} (henceforth \emph{tree}), over a set $X$, determined by a signature $\Sigma$ of effect operators,  is a labelled tree of possibly infinite depth, whose nodes have the following possible forms:
	\begin{enumerate}
		\item A leaf node labelled $\bot$ (representing divergence).
		\item A leaf node labelled $\top$ (representing success or termination).
		\item A leaf node labelled $\leaf{x}$ where $x \in X$.
		\item A node labelled $\EfOp$ with children $t_1,\dots, t_{m}$, when the operator $\EfOp \in \Sigma$ has arity $\Arr{\EfOp} = m$. In this case, we write the subtree at that node as $\EfOp\langle t_1,\dots,t_m\rangle$.
		\item A node labelled $\EfOp$ with an infinite sequence $t_0,t_1,\dots$ of children, when the operator $\EfOp \in \Sigma$ has arity $\Arr{\EfOp} = \Nat$. We write the subtree at that node as $\EfOp\langle m \mapsto t_m \rangle$ (we may use this notation for the nodes described in point 4 too).
	\end{enumerate}
\end{defi}
\noindent
This definition varies slightly from effect trees used in \cite{modal_journal,Quantitative}, with the addition of a top element $\top$.

We define a functor $\Trees{_{\Sigma}(-)}$ on the category of sets, sending each set $X$ to the set of trees $\Trees{_{\Sigma}X}$ over $X$ determined by $\Sigma$, and sending each function $f: X \to Y$ to the function $\Trees{f}: \Trees{_{\Sigma}X} \to \Trees{_{\Sigma}Y}$ replacing each leaf $\leaf{x}$ of its input by $\leaf{f(x)}$.
We will henceforth write $\Trees{X}$ instead of $\Trees{_{\Sigma}X}$, leaving the underlying signature $\Sigma$ implicit.
The functor $\Trees{(-)}$ determines a monad $(\Trees{ }, \eta, \mu)$, where $\eta(x) = \leaf{x}$, and $\mu$ flattens a \emph{double-tree} $d \in \Trees{\Trees{X}}$ into a tree $\mu d \in \Trees{X}$ by replacing each leaf $\leaf{t}$ of $d$ by $t$ as a subtree.
Both $\eta$ and $\mu$ are natural transformations, satisfying the monad laws.
For $f: X \to \Trees{Y}$, define $f^* := \mu \circ \Trees{(f)} : \Trees{X} \to \Trees{Y}$.

Given a preorder $(X,\leq_X)$ we define an order on $\Trees{X}$ coinductively according to the following rules:
\begin{enumerate}
	\item $\forall t \in \Trees{X}. \ \bot \leq_{\Trees{X}} t \leq_{\Trees{X}} \top$.
	\item $\forall x, y \in X. \ x \leq_X y \ \ \implies \ \ \leaf{x} \leq_{\Trees{X}} \leaf{y}$.
	\item $\EfOp\langle m \to t_m \rangle \leq_{\Trees{X}} r \ \ \implies \ \ \exists r_1,r_2,\dots \in \Trees{X} \ \text{s.t.}  \ (r = \EfOp\langle m \to r_m \rangle \wedge \forall m. \ t_m \leq_{\Trees{X}} r_m)$.
\end{enumerate}

An order is \emph{$\omega$-complete} if it contains the supremum (limit) of any ascending sequence of elements.
If $X$ is an $\omega$-complete preorder, then $\Trees{X}$ is an $\omega$-complete preorder.
Note that if $X$ has the discrete order, it is $\omega$-complete, and hence $\Trees{X}$ is $\omega$-complete, and its order is specified by rules 1 and 3 only.

\subsection{Equations and Inequations}

We use the natural numbers $\Nat$ to describe a countable set of variables, and trees $\Trees{\Nat}$ as the set of possibly infinite \emph{algebraic expressions}.
An \emph{algebraic equation} is simply the assertion that two expressions $a,b \in \Trees{\Nat}$ are equal `$a = b$', and an \emph{algebraic inequation} is the assertion that two expressions $a,b \in \Trees{\Nat}$ are ordered `$a \leq b$'.
Both such statements can be seen as elements of $(\Trees{\Nat})^2$.

We study \emph{algebraic relations} $\IEQ \subseteq (\Trees{\Nat})^2$ containing such assertions, in particular inequations, and write $\Ias{a}{b}$ for $(a,b) \in \IEQ$.
We study properties of $\IEQ$, given e.g. in \cite{op_meta,LopezSimpson}.
\begin{enumerate}
	\item[R.] $\IEQ$ is \emph{reflexive} if for any $a \in \Trees{\Nat}$, then $\Ias{a}{a}$.
	\item[T.] $\IEQ$ is \emph{transitive} if for any $a, b, c \in \Trees{\Nat}$, $\Ias{a}{b} \wedge \Ias{b}{c} \implies \Ias{a}{c}$.
\end{enumerate}
If $\IEQ$ satisfies R and T, then it is a preorder. The next two properties discuss substituting trees for variables in the algebraic expressions, with compositionality from \cite{op_meta}.
\begin{enumerate}
	\item[S.] $\IEQ$ is \emph{substitutional} if $\forall (a,b) \in \IEQ$ and $f: \Nat \to \Trees{\Nat}$, \ $\Ias{f^*(a)}{f^*(b)}$.
	\item[C.] $\IEQ$ is \emph{compositional} if $\forall (a,b) \in \IEQ$ and $f,g: \Nat \to \Trees{\Nat}$, s.t. $\forall n \in \Nat. \Ias{f(n)}{g(n)}$, \ $\Ias{f^*(a)}{g^*(b)}$.
\end{enumerate}
Note that reflexivity and compositionality together imply substitutionality.
We consider two more properties, concerning the order $\leq_{\Trees{\Nat}}$ on $\Trees{\Nat}$.
\begin{enumerate}
	\item[O.] $\IEQ$ is \emph{ordered} if $\forall a,b \in \Trees{\Nat}. \ a \leq_{\Trees{\Nat}} b \implies \Ias{a}{b}$.
	\item[A.] $\IEQ$ is \emph{admissible} if for any two increasing sequences $\{a_n\}_{n \in \Nat}$ and $\{b_n\}_{n \in \Nat}$, if $(\forall n \in \Nat. \ \Ias{a_n}{b_n})$ then $\Ias{\bigvee_{n \in \Nat} a_n}{\bigvee_{n \in \Nat} b_n}$ \ (here $\bigvee_{n \in \Nat} c_n$ is the limit/sumpremum of the sequence $\{c_n\}_{n \in \Nat}$).
\end{enumerate}
Note that if $\IEQ$ is ordered, then $\Ias{\bot}{a} \ \Ias{}{\top}$  for any $a \in \Trees{\Nat}$. Moreover, if $\IEQ$ is ordered, then it is reflexive.
We call an algebraic relation $\IEQ$ \emph{complete} if it satisfies all of the six properties given above, though as noted it is enough to satisfy T, C, O and A. Note that $\leq_{\Trees{\Nat}}$ is a complete algebraic relation.

Given a set of axioms $\Axioms \subseteq (\Trees{\Nat})^2$, we define the resulting algebraic relation $\IEQ(\Axioms)$ as the transitive, compositional and admissible closure of the relation given by $\Axioms \cup (\leq_{\Trees{\Nat}})$.
As such, $\IEQ{(\Axioms)}$ is the smallest complete relation containing $\Axioms$. See \cite{LopezSimpson} for some more details on axiomatically defined preorders.

We consider the empty set $\Zero$ as a subset of $\Nat$ given by $\Zero = \{\}$.
As such, we see $\Trees{\Zero}$ as the subset of $\Trees{\Nat}$ containing algebraic expressions without variables, only having $\top$ and $\bot$ as leaves.
In this paper, $\Trees{\Zero}$ takes the place of the set unit type trees from \cite{modal_journal} as the basis for studying effects\footnote{There is a bijection between the two sets, with the $\top$ leaf corresponding to the unit leaf. However, this bijection does not preserve the order.}.

We consider one more property for algebraic relations.
\begin{defi}
	$\IEQ$ is \emph{base-valued} if for any $a,b \in \Trees{\Nat}$:
	$$(\forall f: \Nat \to \Trees{\Zero}. \ \Ias{f^*(a)}{f^*(b)}) \ \implies \ \Ias{a}{b} \enspace . $$
\end{defi}
The property asserts that the algebraic relation $\IEQ$ is completely specified by its subset $\IEQ \cap (\Trees{\Zero})^2$, which we call the \emph{base relation} $\IU$.

\subsection{Effect examples}\label{sub:exa1}
We look at some examples of effects and their algebraic operations.
Moreover, we will specify the usual axiomatic equations given in the literature (e.g. in \cite{PlotkinPower02,Domains}). For clarity, we will use variables $x,y,z,\dots$ instead of numbers when writing elements of $\Trees{\Nat}$, and we will often leave out the leaf-notation, writing $x$ instead of $\leaf{x}$.
For each example, $\IEQ(\Axioms)$ turns out to be base-valued, but we omit the proofs.

\begin{exa}[Nondeterminism]\label{exa:nond1}
	
	We first consider the example of nondeterminism, where the effect signature $\Sigma$ contains a single algebraic effect operator `$\ENor$' of arity $\Arr{\ENor} =  2$.
	This operator chooses between two possible continuations in a completely unpredictable manner, under control of a scheduler which makes choices according to some unknown decision process.
	Because of its unpredictable nature, no probability can be associated to the choices.
	As such, its equational axioms are given by idempotency, symmetry  and associativity:
	\vspace{-1mm}
	$$\ENor(x,x) = x, \qquad \ENor(x,y) = \ENor(y,x), \qquad  \ENor(x,\ENor(y,z)) = \ENor(\ENor(x,y),z).$$
\end{exa}

\begin{exa}[Probability]\label{exa:prob1}
	We consider the example of probability, with one algebraic effect operator `$\EPor$' with arity $2$, which chooses between two continuations randomly, by fair choice.
	In this case, the equational axioms are given by idempotency, symmetry together with two more axioms:
	\vspace{-1mm}
	\begin{align*}
	& \EPor(x,x) = x, \qquad \EPor(x,y) = \EPor(y,x)\\
	& \EPor(\EPor(x,y),\EPor(z,w)) = \EPor(\EPor(x,z),\EPor(y,w)), \qquad  \mu x.\EPor(y,x) = y.
	\end{align*}
	Here, $\mu x.\EPor(y,x)$ stands for the infinite tree $t$ such that $t = \EPor(y,t)$. 
\end{exa}

\begin{exa}[Global Store]\label{exa:glob1}
	We consider a global memory location which contains some natural number.
	Our effect signature $\Sigma$ contains a lookup operator `$\ELo$' with arity $\Nat$, which looks up the stored natural number and continues the computation accordingly, and for each $n \in \Nat$ we have an update operator `$\EUp{n}$' which updates the stored number to an $n$ (this can be generalised to multiple store locations).
	We have the following equations as axioms, ranging over natural numbers $n,m \in \Nat$:
	\vspace{-1mm}
	\begin{align*}
	& \EUp{n}(\EUp{m}(x)) = \EUp{m}(x), \qquad \ELo(m \mapsto x) = x, \qquad \EUp{n}(\ELo(m \mapsto x_m)) = x_n,\\
	& \ELo(m \mapsto \EUp{m}(x_m)) = \ELo(m \mapsto x_m).
	\end{align*}
\end{exa}

\begin{exa}[Exception catching]\label{exa:exep1}
	This example is similar to the algebraic description of the jump effect from \cite{Fiore}.
	We consider a set of exceptions $\Exc$, and for each $e \in \Exc$ we have an operator `$\EEx{e}$' of arity $0$ raising the exception, and an operator `$\ECa{e}$' of arity $2$ catching that exception.
	The computation $\ECa{e}(P,Q)$ will execute the computation $P$, and if the exception $e$ is raised by $P$, it continues by executing the computation given by $Q$.
	We consider the following axiomatic equations:
	\vspace{-1mm}
	\begin{align*}
	& \ECa{e}(\EEx{e},x) = x, \qquad \ECa{e}(\EEx{d},x) = \EEx{d} \ \ \text{if} \ e \neq d, \quad \ECa{e}(x,x) = x,\\ 
	& \ECa{e}(\ECa{e}(x,y),z) = \ECa{e}(x,\ECa{e}(y,z)), \qquad \ECa{e}(\bot,x) = \bot, \qquad \ECa{e}(\top,x) = \top.
	\end{align*}
\end{exa}

\begin{exa}[Input]\label{exa:inpu1}
	We consider the situation in which a computation may ask for a binary input from the user of the computer.
	This is modelled using a single operation `$\EIn$' of arity two, where $\EIn(P,Q)$ is the computation which asks a binary input, and continues with $P$ if the input is $0$, and $Q$ if the input is $1$.
	In this example, the entity giving the inputs can keep track of what choices are made.
	As a result, we will not assume any axioms, since any two different trees of $\Trees{\Nat}$ can be distinguished by testing their evaluation with a particular sequence of inputs.
\end{exa}

\begin{exa}[Cost]\label{exa:cost1}
	We consider the situation in which we associate a cost to computation, for instance energy, time, or monetary cost necessary to evaluate a program.
	We consider a single \emph{tick} operation `$\ETi$' with arity one, where $\ETi(P)$ evaluates $P$ after a unit of cost has been paid.
	Elements of $\Trees{\Nat}$ are given by a sequence of ticks, which is either infinite, or results in $\bot$, $\top$, or a natural number.
	We consider one axiom, asserting that cheaper is better:
	\vspace{-1mm}
	$$ \ETi(x) \ \leq \ x \ . $$
	As a consequence, we can show that $\bot \leq \ETi(\bot) \leq \bot$, and hence $\ETi(\bot) = \bot$.
	Using admissibility, we can prove that the algebraic expression given by an infinite sequence of ticks is equal to $\bot$.
\end{exa}

\begin{exa}[Nondeterminism + Probability]\label{exa:comb1}
	As a last example, we look at a combination of effects.
	As signature we take $\Sigma := \{\ENor,\EPor\}$ with two binary operators, and we assume the axiomatic equations of the two effects from Examples~\ref{exa:nond} and~\ref{exa:prob} hold.
	Moreover, we state the following interaction law:
	$$\EPor(x,\ENor(y,z)) \ = \ \ENor(\EPor(x,y), \EPor(x,z)) \ .$$
\end{exa}

\section{Eilenberg-Moore Algebras}\label{Sec:EM}

On the opposite side of equations, we have distinctions. We will use Eilenberg-Moore algebras to specify tests on algebraic expressions as done in \cite{Quantitative}.
If two expressions give us a different result for a test, we consider them to be distinct.

\begin{defi}
	Given a monad $(M,\eta,\mu)$, an \emph{Eilenberg-Moore algebra} (henceforth EM-algebra) is a morphism $\alpha: M A \to A$ on some carrier object $A$, such that the following two diagrams commute:
	$$ \xymatrix{ 
		A \ar@{->}[r]_{\eta_A} & M A \ar@{->}[d]_{\alpha} & & M M A \ar@{->}[r]_{M \alpha} \ar@{->}[d]_{\mu_A} & M A \ar@{->}[d]_{\alpha}\\
		& A \ar@{<-}[ul]^{\textit{id}}  & & M A \ar@{->}[r]_{\alpha} & A
	}$$
\end{defi}

Given a preorder $(\Arb, \fimp)$, and an algebra $\alpha: \Trees{\Arb} \to \Arb$ on the tree functor $\Trees{(-)}$, we define a relation $\Arel \ \subseteq (\Trees{\Nat})^2$ as follows:
$$ a \Arel b \quad \iff \quad \forall h: \Nat \to \Arb. \alpha(\Trees{(h)}(a)) \fimp \alpha(\Trees{(h)}(b))$$
We say that $\alpha$ \emph{complements} $\IEQ$ if $\Arel$ coincides with $\Ias{}{}$.
The algebra $\alpha$ \emph{complements} a set of axiomatic equations $\Axioms$, if $\alpha$ complements $\IEQ{(\Axioms)}$.

Suppose $\alpha$ complements $\Axioms$.
For any two algebraic expressions $a,b \in \Trees{\Nat}$, it is either possible to show that $a \sqsubseteq_{\IEQ{(\Axioms)}} b$ using the axioms from $\Axioms$ and proof rules such as compositionality and admissibility, or it is possible to show that $a \not\sqsubseteq_{\IEQ{(\Axioms)}} b$ using the EM-algebra $\alpha$ together with some test $h: \Nat \to \Arb$.
As such, we have both a method for showing equivalence, and for showing inequivalence.

Note that $\Arel$ is reflexive and transitive. We look at some other general results.
\begin{lem}
	If $\alpha$ is an EM-algebra, then $\Arel$ is substitutional.
\end{lem}
\begin{proof}
	Assume $a \Arel b$, and take $f: \Nat \to \Trees{\Nat}$, we need to show that $f^*(a) \Arel f^*(b)$.
	Let $h: \Nat \to \Arb$, then $(\alpha \circ \Trees{h} \circ f): \Nat \to \Arb$, and hence
	$\alpha(\Trees{(\alpha \circ \Trees{h} \circ f)}(a)) \ \fimp \ \alpha(\Trees{(\alpha \circ \Trees{h} \circ f)}(b))$.
	Note that the following diagram commutes:
	\begin{equation}\label{dia:lift}
	\xymatrix{
		\Trees{\Nat} \ar@{->}[r]_{T f} \ar@{->}[dr]_{f^*} & \Trees{\Trees{\Nat}} \ar@{->}[r]_{\Trees{\Trees{h}}} \ar@{->}[d]^{\mu_{\Nat}} & \Trees{\Trees{\Arb}} \ar@{->}[r]_{\Trees{\alpha}} \ar@{->}[d]_{\mu_{\Arb}} & \Trees{\Arb} \ar@{->}[d]_{\alpha}\\
		& \Trees{\Nat} \ar@{->}[r]^{\Trees{h}} & \Trees{\Arb} \ar@{->}[r]^{\alpha} & \Arb
	}
	\end{equation}
	Hence $(\alpha \circ \Trees{h} \circ f^*) \ (a) \ \fimp \ (\alpha \circ \Trees{h} \circ f^*) \ (b)$.
	So we have the desired result.
\end{proof}

We call $\alpha$ \emph{monotone} if it preserves order: $\forall a, b \in \Trees{\Arb}. \ a \leq_{\Trees{\Arb}} b \implies \alpha(a) \fimp \alpha(b)$.

\begin{lem}
	If $\alpha$ is a monotone EM-algebra, then $\Arel$ is ordered and compositional.
\end{lem}
\begin{proof}
	Let $a \leq_{\Trees{\Nat}} b$ and $h: \Nat \to \Arb$, then $\Trees{(h)}(a) \leq_{\Trees{\Arb}} \Trees{(h)}(b)$ and hence by monotonicity of $\alpha$ we get $\alpha(\Trees{(h)}(a)) \fimp \alpha(\Trees{(h)}(b))$.
	
	Assume $a \Arel b$, and take $(f,g): \Nat \to \Trees{\Nat}$ s.t. $\forall n. \ f(n) \Arel g(n)$, we need to show that $f^*(a) \Arel g^*(b)$.
	Because of substitutivity, $g^*(a) \Arel g^*(b)$, so with transitivity it is sufficient to show that $f^*(a) \Arel g^*(a)$.
	
	Let $h: \Nat \to \Arb$, so for all $n \in \Nat$, $\alpha(\Trees{(h)}(f(n))) \fimp \alpha(\Trees{(h)}(g(n)))$. Hence
	$\Trees{(\alpha \circ \Trees{h} \circ f)}(a) \ \leq_{\Trees{\Arb}} \ \Trees{(\alpha \circ \Trees{h} \circ g)}(a)$, and since $\alpha$ is monotone:
	$$
	\alpha(\Trees{(\alpha \circ \Trees{h} \circ f)}(a)) \ \fimp \ \alpha(\Trees{(\alpha \circ \Trees{h} \circ g)}(a)) \enspace . \\
	$$
	Using diagram (\ref{dia:lift}) from the previous lemma, we conclude that $(\alpha \circ \Trees{h} \circ f^*) \ (a) \ \fimp \ (\alpha \circ \Trees{h} \circ g^*) \ (a)$, and hence $f^*(a) \Arel g^*(a)$. We conclude that $\Arel$ is compositional.
\end{proof}

Last but not least, we establish a sufficient condition for admissibility. Note that $\alpha$ is a morphism in the category of $\omega$-cpos precisely if $\Arb$ is $\omega$-complete and $\alpha$ preserves limits of ascending sequences.

\begin{lem}
	If $\alpha: \Trees{\Arb} \to \Arb$ is an EM-algebra in the category of $\omega$-cpos, then $\Arel$ is admissible.
\end{lem}
\begin{proof}
	Let $\{a_i\}_{i \in \Nat}$ and $\{b_i\}_{i \in \Nat}$ be ascending sequences of trees from $\Trees{\Nat}$ such that $\forall i \in \Nat. \ a_i \Arel b_i$.
	Let $h: \Nat \to \Arb$, then $\forall i \in \Nat. \ \alpha(\Trees{(h)}(a_i)) \fimp \alpha(\Trees{(h)}(b_i))$ and hence: 
	$$\alpha(\Trees{(h)}(\bigvee_{i \in \Nat}a_i)) = \bigvee_{i \in \Nat}\alpha(\Trees{(h)}(a_i)) \fimp \bigvee_{i \in \Nat}\alpha(\Trees{(h)}(b_i)) = \alpha(\Trees{(h)}(\bigvee_{i \in \Nat}a_i)) \enspace .$$
	We conclude that $\bigvee_{i \in \Nat}a_i \Arel \bigvee_{i \in \Nat}b_i$, so $\Arel$ is admissible.
\end{proof}

\subsection{From Equations to EM-algebras}\label{sec:fromto}
We specified how we can extract a relation $\Arel$ on $\Trees{\Nat}$ from an EM-algebra $\alpha$. If this algebra forms a morphism in the category of $\omega$-cpos, the resulting relation is complete.
We will now go in the other direction, extracting an algebra from a relation on $\Trees{\Nat}$ in a novel way. 
In particular, we will formulate an EM-algebra using the relation $\Ias{}{}$ specified by $\IEQ$.

We denote by $\overline{\mathcal{R}}$ the largest symmetric subset of a relation $\mathcal{R}$. 	Remember that $\IU = \IEQ \cap (\Trees{\Zero})^2$. 
\begin{defi}
	The \emph{value space} of $\IEQ$ is given by $\ClIU := \{\clos{a} \mid a \in \Trees{\Zero}\}$, where $\clos{a} := \{b \in \Trees{\Zero} \mid a \ \overline{\IU} \ b\}$.
\end{defi}
We have a function $\clos{-}: \Trees{\Zero} \to \ClIU$ defined by $a \mapsto \clos{a}$.
Note that by transitivity, $(b \in \clos{a}) \iff ({(\clos{a} \cap \clos{b})} \neq \emptyset) \iff (\clos{a} = \clos{b}) \iff (a \ \overline{\IU} \ b)$.
We define an order $\fimp$ on $\ClIU$ where $(\clos{a} \fimp \clos{b}) :\Leftrightarrow (a \ \IU \ b)$, which by transitivity of $\IU$ is well-defined.

A \emph{choice function} for $\ClIU$ is a function $c: \ClIU \to \Trees{\Zero}$ such that for all $S \in \ClIU$, $c(S) \in S$. Note by the above properties that for all $S \in \ClIU$, $\clos{c(S)} = S$.
We specify an algebra $\alpha_c: \Trees{\ClIU} \to \ClIU$ as the function which makes the following diagram commute:
$$\xymatrix{
	\Trees{\ClIU} \ar@{.>}[r]_{\alpha_c} \ar@{->}[d]_{T c} & \ClIU \\
	\Trees{\Trees{\Zero}} \ar@{->}[r]_{\mu_{\Zero}} & \Trees{\Zero} \ar@{->}[u]_{\clos{-}}
}$$

\begin{lem}
	If $\IEQ$ is reflexive, transitive and compositional, then for any two choice functions $c$ and $d$, $\alpha_c$ is equal to $\alpha_d$.
\end{lem}
\begin{proof}
	Note that for any $S \in \ClIU$, $c(S) \in S$ and $d(S) \in S$, hence $c(S) \ \overline{\IU} \ d(S)$. Take some $t \in \Trees{\ClIU}$.
	Note that the number of leaves of $t$ is countable, hence we can find a tree $k \in \Trees{\Nat}$ and a function $f: \Nat \to \ClIU$ such that $t = T(f)(k)$.
	Since $k \ \IEQ \ k$ by reflexivity, and for each $n \in \Nat$, $c(f(n)) \ \overline{\IU} \ d(f(n))$, it holds by compositionality that $(c \circ f)^*(k) \ \overline{\IU} \ (d \circ f)^*(k)$.
	Hence $\alpha_c(t) = \clos{\mu_{\Zero}(T(c)(t))} = \clos{(c \circ f)^*(k)} = {\clos{(d \circ f)^*(k)}} = \clos{\mu_{\Zero}(T(d)(t))} = \alpha_d(t)$.
\end{proof}

Hence, $\alpha_c$ is invariant under choice of $c$.
From now on, we will fix a choice function $c$, and simply write $\alpha$ for $\alpha_c$.
We will show that $\alpha$ is an EM-algebra. First, we establish a useful lemma.

\begin{lem}\label{lem:comm}
	If $\IEQ$ is reflexive, transitive and compositional, then the following diagram commutes:
	$$\xymatrix{
		\Trees{\ClIU} \ar@{->}[r]_{\alpha_c}  & \ClIU \\
		\Trees{\Trees{\Zero}} \ar@{->}[r]_{\mu_{\Zero}} \ar@{->}[u]^{T\clos{-}} & \Trees{\Zero} \ar@{->}[u]_{\clos{-}}
	}$$
\end{lem}
\begin{proof}
	This has a similar proof to the previous lemma. Take $t \in TT\Zero$, and define $k \in \Trees{\Nat}$ and $f: \Nat \to \Trees{\Zero}$ such that $t = T(f)(k)$.
	For all $n \in \Nat$, $f(n) \ \overline{\IU} \ c(\clos{f(n)})$ holds, hence by compositionality, 
	$f^*(k) \ \overline{\IU} \ {(c \circ \clos{-} \circ f)^*(k)}$.
	So
	$\clos{\mu_{\Zero}(t)} = \clos{f^*(k)} = \clos{(c \circ [-] \circ f)^*(k)} = \clos{(c^*\circ T([-]) \circ T(f))(k)} = {\clos{c^*\circ T([-]) (t)}}  = \alpha(T([-])(t))$.
\end{proof}

\begin{pro}
	If $\IEQ$ is reflexive, transitive and compositional, then $\alpha$ is an EM-algebra.
\end{pro}
\begin{proof}
	We use the monad laws together with definition of $\alpha$ and Lemma \ref{lem:comm} to observe that the following diagrams commute:
	$$\xymatrix{
		\ClIU \ar@/_4pc/[dddrrr]_{\textit{id}} \ar@/_2pc/[ddrr]_{c} \ar@{->}[rrr]^{\eta_{\ClIU}} \ar@{->}[dr]^{c} & & & T\ClIU \ar@{->}[ddd]^{\alpha} \ar@{->}[dl]^{\Trees{c}}\\
		& T\Zero \ar@{->}[r]^{\eta_{T\Zero}}  \ar@{->}[dr]_{\textit{id}} & TT\Zero \ar@{->}[d]^{\mu_{\Zero}} & \\
		& & T\Zero \ar@{->}[dr]^{\clos{-}} & \\
		& & & \ClIU
	} \qquad \qquad
	\xymatrix{
		T T \ClIU \ar@{->}[rrr]^{T\alpha} \ar@{->}[dr]^{T T c} \ar@{->}[ddd]^{\mu_{\ClIU}} & & & T \ClIU \ar@{->}[ddd]^{\alpha}\\
		& T T T \Zero \ar@{->}[r]^{T\mu_{\Zero}} \ar@{->}[d]^{\mu_{T \Zero}} & T T \Zero \ar@{->}[ru]^{T \clos{-}} \ar@{->}[d]^{\mu_{\Zero}} & \\
		& T T \Zero \ar@{->}[r]^{\mu_{\Zero}} & T \Zero \ar@{->}[dr]^{\clos{-}} & \\
		T \ClIU \ar@{->}[rrr]^{\alpha} \ar@{->}[ur]^{T c} & & &  \ClIU
	}$$
\end{proof}

\begin{lem}
	If $\IEQ$ is reflexive, transitive and compositional, then for all $a,b \in \Trees{\Nat}$, $\Ias{a}{b} \implies a \sqsubseteq_{\alpha} b$.
\end{lem}
\begin{proof}
	Assume that $\Ias{a}{b}$, and let $f : \Nat \to \ClIU$, then $c \circ f : \Nat \to \Trees{\Zero}$. 
	So by compositionality, ${(c \circ f)^*(a)}$ $\Ias{}{} \ (c \circ f)^*(b)$, hence $(c \circ f)^*(a) \ \IU \ (c \circ f)^*(b)$, so we conclude that:
	$\alpha(T(f)(a)) = {\clos{(c \circ f)^*(a)}} \fimp \clos{(c \circ f)^*(b)} = \alpha(T(f)(b))$.
\end{proof}

To prove that the constructed EM-algebra complements the algebraic relation, we use a relation lifting operation.
Given a relation $\mathcal{R} \subseteq X \times Y$, we define the lifted relation $\rell{\mathcal{R}} \subseteq \Trees{X} \times \Trees{Y}$ coinductively as follows:
\begin{enumerate}
	\item $\bot \ \rell{\mathcal{R}} \ t \implies t = \bot$, \qquad $\top \ \rell{\mathcal{R}} \ t \implies t = \top$
	\item $\leaf{x} \ \rell{\mathcal{R}} \ t \implies \exists y \in Y. \ t = \leaf{y} \wedge x \ \mathcal{R} \ y$.
	\item $\EfOp\langle m \mapsto t_m \rangle \ \rell{\mathcal{R}} \ r \implies \exists \{r_m\}_{m \in \Nat}. \ r = \EfOp\langle m \mapsto r_m \rangle \wedge (\forall m. \ t_m \ \rell{\mathcal{R}} \ r_m)$.
\end{enumerate}
This can be seen as the functorial lifting of relations, and has the following two properties:
\begin{itemize}
	\item $\forall x \in X, \forall y \in Y. \ x \ \mathcal{R} \ y \implies \eta(x) \ \rell{\mathcal{R}} \ \eta(y)$.
	\item $\forall a \in \Trees{\Trees{X}}, b \in \Trees{\Trees{Y}}. \ a \ \rell{\rell{\mathcal{R}}} \ b \implies \mu a \ \rell{\mathcal{R}} \ \mu b$.
\end{itemize}

\begin{pro}\label{pro:complement}
	If $\IEQ$ is reflexive, transitive, compositional and base-valued, then $\alpha$ complements $\IEQ$.
\end{pro}
\begin{proof}
	Assume that $a \sqsubseteq_{\alpha} b$, we prove that $\Ias{a}{b}$ using that $\IEQ$ is base-valued, in order to conclude using the previous lemma that $\alpha$ complements $\IEQ$.
	
	Let $f: \Nat \to \Trees{\Zero}$, then $(\clos{-} \circ f) : \Nat \to \ClIU$, hence $\alpha(T(\clos{-} \circ f)(a)) \fimp \alpha(T(\clos{-} \circ f)(b))$.
	So we get $[\mu (T(c \circ \clos{-} \circ f)(a))] \fimp [\mu (T(c \circ \clos{-} \circ f)(b))]$, and hence $\mu (T(c \circ \clos{-} \circ f)(a)) \ \IU \ \mu (T(c \circ \clos{-} \circ f)(b))$.
	Note that for all $t \in \Trees{\Zero}$, $t \ \overline{\IU} \ (c \circ \clos{-})(t)$.
	
	Hence for any $k \in \Trees{\Nat}$, $T(f)(k) \ \rell{\overline{\IU}} \ {T(c \circ \clos{-} \circ f)(k)}$, so by compositionality, we derive that $f^*(k) = \mu(T(f)(k)) \ \overline{\IU} \ \mu(T(c \circ \clos{-} \circ f)(k))$.
	We conclude that:
	$$f^*(a) \ \overline{\IU} \ \mu(T(c \circ \clos{-} \circ f)(a)) \ \IU \ \mu(T(c \circ \clos{-} \circ f)(b)) \ \overline{\IU} \ f^*(b) \enspace .$$
	So by transitivity, $f^*(a) \ \IU \ f^*(b)$.
	Hence by base-valuedness, $\Ias{a}{b}$.
\end{proof}

\subsection{EM-algebras for the examples}
We look at the examples of effects given in Subsection~\ref{sub:exa1}, and study what value spaces $\clos{\IU}$ and EM-algebras $\alpha: \Trees{\clos{\IU}} \to \clos{\IU}$ they generate.
For every example, the constructed EM-algebra satisfies the following rules:
$\alpha(\bot) = \clos{\bot}$, $\alpha(\top) = \clos{\top}$, $\alpha(\leaf{a}) = a$ for any $a \in \clos{\IU}$, and $\alpha(\bigvee_{n \in \Nat} t_n) = \bigvee \alpha(t_n)$ for any sequence $t_0 \leq_{\Trees{\Zero}} t_1 \leq_{\Trees{\Zero}} t_2 \leq_{\Trees{\Zero}} \dots$.
To complete the definition of the algebras in the following examples, we will specify their behaviour over algebraic effect operators. These \emph{local functions} together with the above properties uniquely characterise the morphism $\alpha$.

\begin{exa}[Nondeterminism]\label{exa:nond}
	
	The first example is nondeterminism with the binary operation $\ENor$.
	We study $\IU$ as derived from the induced equational theory $\IEQ(\Axioms)$, where $\Axioms$ is given by the equations from Example~\ref{exa:nond}.
	We get three elements of $\ClIU$, each denoting an equivalence class of $\IU$.
	\begin{enumerate}
		\item[$\bot$.] Any tree $t \in \Trees{\Zero}$ without a $\top$-leaf is equivalent to $\bot$.
		\item[$\Diamond$.] Any tree $t \in \Trees{\Zero}$ with at least one $\top$-leaf, which moreover either has a $\bot$ leaf or is infinite, is equivalent to $\ENor(\top,\bot)$.
		\item[$\top$.] Any tree $t \in \Trees{\Zero}$ which is finite and only has $\top$-leaves, is equivalent to $\top$.
	\end{enumerate}
	If we write $\Diamond$ for $\ENor(\top,\bot)$, the ordered set $\ClIU$ is given by $\{\clos{\bot} \fimp \clos{\Diamond} \fimp \clos{\top}\}$.
	
	We give some informal arguments for the above observations. Note first that:
	$$\bot \leq_{\Trees{\Zero}} \ENor(\top,\bot) \leq_{\Trees{\Zero}} \top ,$$
	which proves the given ordering.
	Now consider a finite tree $t \in \Trees{\Zero}$. If $t$ only has one type of leaf, consequetive application of the idempotency axiom reduces that tree to just that leaf. 
	If $t$ has both $\top$ and $\bot$ leaves, applying symmetry and transitivity can change the tree to one of the form $\ENor(l,r)$, where $l$ only has $\top$ leaves, and $r$ only $\bot$-leaves.
	So $t$ can be reduced to $\Diamond$ with idempotency.
	Now consider an infinite tree $t \in \Trees{\Zero}$, and let $t_0 \leq_{\Trees{\Zero}} t_1 \leq_{\Trees{\Zero}} t_2 \leq_{\Trees{\Zero}} \dots$ be a sequence of finite trees approximating $t$.
	If $t$ only has $\bot$ leaves, all finite approximations only have $\bot$ leaves, and each $t_i$ is equal to $\bot$. So by admissibility, $t \ = \ \bigvee_{n \in \Nat} t_n \ \overline{\IU} \ \bot$.
	If $t$ has at least one $\top$ leaf, all finite approximations have at least one $\bot$-leaf (by studying $\leq_{\Trees{\Zero}}$).
	Moreover, there must be an $n \in \Nat$ such that for all $m \geq n$, $t_m$ has a $\top$-leaf, so $t_m \ \overline{\IU} \ \Diamond$. Hence $t \ \overline{\IU} \ \bigvee_{n \in \Nat} t_n \ \overline{\IU} \ \Diamond$.
	We conclude that all trees must be equal to either $\bot$, $\top$ or $\Diamond$.
	
	The induced EM-algebra $\alpha: \Trees{\clos{\IU}} \to \clos{\IU}$ corresponds to the algebra operation given in \cite{bat}, with $\alpha(\ENor(a,a)) = a$ and $\alpha(\ENor(a,b)) = \clos{\Diamond}$ for any $a,b \in \clos{\IU}$ such that $a \neq b$.
\end{exa}

\begin{exa}[Probability]\label{exa:prob}
	We study $\IEQ(\Axioms)$ resulting from Example~\ref{exa:prob1}.
	Consider the real number interval $[0,1]$, and the function $\mathsf{P}: \Trees{\Zero} \to [0,1]$ satisfying the following rules $\mathsf{P}(\bot) = 0$, $\mathsf{P}(\top) = 1$, $\mathsf{P}(\EPor(l,r)) = (\mathsf{P}(l) + \mathsf{P}(r))/2$, and $\mathsf{P}(\bigvee_{n \in \Nat} t_n) = \lim_{n \to \infty}\mathsf{P}(t_n)$ for any ascending sequence $t_0 \leq_{\Trees{\Zero}} t_1 \leq_{\Trees{\Zero}} t_2 \leq_{\Trees{\Zero}} \dots$.
	For any two elements $a,b \in \Trees{\Zero}$, $a \ \IU \ b$ holds if and only if $\mathsf{P}(a) \leq \mathsf{P}(b)$. Moreover, $\mathsf{P}$ is surjective\footnote{Since rationals with power 2 denominators are dense in the real numbers.}.
	Hence $\ClIU$ can be expressed as $[0,1]$, where $\clos{-}: \Trees{\Zero} \to \ClIU$ is given by $\mathsf{P}$.
	The induced EM-algebra $\alpha : \Trees{[0,1]} \to [0,1]$ calculates the \emph{expected result}, where $\alpha(\EPor(a,b)) = (a+b)/2$.
\end{exa}

	\begin{exa}[Global Store]\label{exa:glob}
	The above two examples are standard in the literature, and do not explicitly use that element $\top$ is the top element of $\Trees{\Zero}$.
	In the case of global store however, this fact is important.
	We study $\IEQ = \IEQ(\Axioms)$ resulting from Example~\ref{exa:glob1}.
	
	Note that $\EUp{n}(\bot) \ \leq_{\Trees{\Nat}} \ \EUp{n}(\EUp{m}(\bot)) \ \IEQ \ \EUp{m}(\bot)$.
	Hence for any two natural numbers $n, m \in \Nat$, $\EUp{n}(\bot) \ \IEQ \ \EUp{m}(\bot)$ .
	So: 
	$$\bot \ \overline{\IEQ} \ \ELo(m \mapsto \bot) \ \overline{\IEQ} \ \ELo(m \mapsto \EUp{m}(\bot)) \ \overline{\IEQ} \ \ELo(m \mapsto \EUp{n}(\bot)) \ \overline{\IEQ} \ \EUp{n}(\bot) \enspace .$$
	With similar reasoning, $\EUp{n}(\top) \ \overline{\IEQ} \ \top$.
	We derive that for any tree $t \in \Trees{\Nat}$, there is a function $f: \Nat \to (\{\EUp{m}(\leaf{n}) \mid n,m \in \Nat\}\cup\{\bot,\top\})$ such that $t \ \overline{\IEQ} \ \ELo(m \mapsto f(m))$.
	
	Studying $\IU$ in particular, we see that each $a \in \Trees{\Zero}$ is equivalent to $\ELo(m \mapsto f(m))$ for some unique function $f_a: \Nat \to \{\bot,\top\}$.
	For $a,b \in \Trees{\Nat}$, $a \ \IU \ b$ holds if and only if for any $m \in \Nat$, ${(f_a(m) = \top)} \Rightarrow (f_b(m) = \top)$.
	Note moreover that for any function $f: \Nat \to \{\bot,\top\}$, there is an element $a \in \Trees{\Nat}$ such that $f_a = f$.
	So, $\ClIU$ can be expressed as the powerset $\Power{\Nat}$,
	where $\clos{a} := \{n \in \Nat \mid f_a(n) = \top\}$, and the order is given by inclusion.
	We see this powerset as the set of \emph{assertions} on the global state.
	
	The induced EM-algebra $\alpha: \Trees{\mathcal{P}(\Nat)} \to \mathcal{P}(\Nat)$ calculates the \emph{weakest precondition}: it gives the set of starting sates for which the tree reaches a leaf $\leaf{A}$ with a final state $s$ satisfying the assertion $A$.
	\end{exa}

	\begin{exa}[Exception catching]\label{exa:exep}

	We look at $\IU$ resulting from Example~\ref{exa:exep1}.
	Note that any element of $\Trees{\Zero}$ is, under $\IU$, equivalent to one of the following three types of trees: $\bot, \top$ or $\EEx{e}$ for some $e \in \Exc$. The elements are ordered in the following way:
	$$
		\forall e \in \Exc. \qquad \clos{\bot} \qquad \fimp \qquad
	 	\clos{\EEx{e}} \qquad \fimp \quad
		\clos{\top} \enspace .
	$$
	There is no ordering between $\EEx{e}$ and $\EEx{d}$ if $e \neq d$.
	The induced EM-algebra $\alpha$ is defined by $\alpha(\EEx{e}) = \clos{\EEx{e}}$.	
\end{exa}

	\begin{exa}[Input]
	We look at the input effect as given in Example~\ref{exa:inpu1}.
	We did not specify any axioms there, which can be motivated from the perspective of testing distinctions.
	A test of an input program would be checking a series of inputs until either: a) the program terminates successfully (marked by $\top$), or b) the program stops asking for inputs because of divergence (marked by $\bot$).
	Since there are no axioms, $\IU$ \ is given by $\leq_{\Trees{\Zero}}$, $\clos{\IU}$ is given by $\Trees{\Zero}$, and the constructed EM-algebra $\alpha$ is given by the function $\mu_{\Zero}: \Trees{\Trees{\Zero}} \to \Trees{\Zero}$.
	\end{exa}
	
	\begin{exa}[Cost]
	We look at $\IU$ for the cost effect given in Example~\ref{exa:cost1}. Considering the observations made there, we see that $\clos{\IU}$ is given by the set containing:
	$$\clos{\top} \quad \fpmi \quad \clos{\ETi(\top)} \quad \fpmi \quad \clos{\ETi(\ETi(\top))} \quad \fpmi \quad \clos{\ETi(\ETi(\ETi(\top)))} \quad \fpmi \quad \dots \quad \fpmi \quad \clos{\bot} \enspace .$$
	We can represent $\Arb = \ClIU$ as $\Nat_{\infty} = \Nat \cup \{\infty\}$ with reverse order. The constructed EM-algebra is $\alpha : \Trees{\Nat_{\infty}} \to \Nat_{\infty}$, where $\alpha(\ETi(t)) = \alpha(t) + 1$ and $\alpha$ applied to the infinite sequence of ticks gives $\infty$.
	\end{exa}
	
	\begin{exa}[Nondeterminism + Probability]\label{exa:comb}
		We look at the induced EM-algebra for the combination of effects given in Example~\ref{exa:comb1}, which coincides with a description from \cite{LopezSimpson}.
		The value space is given by $\clos{\IU} = \{(a,b) \in [0,1]^2 \mid a \leq b\}$, and the EM-algebra $\alpha$ by $\alpha(\ENor((a,b),(c,d))) = (\min(a,c),\max(b,d))$ and $\alpha(\EPor((a,b),(c,d))) = ((a+c)/2,(b+d)/2)$.
		E.g. $\alpha(\EPor(\top,\ENor(\top,\bot))) = (1/2,1)$.
	\end{exa}
	
	\section{Notes on logic and equivalence}\label{Sec:extra}
	In this section we will look at some more connections with the quantitative logic from \cite{Quantitative} used to specify behavioural equivalence.
	In particular, we will look at two topics discussed in that paper.
	
	\subsection{Relators}\label{sub:relator}
	In proving that the behavioural equivalence given by an EM-algebra is a congruence, the paper \cite{Quantitative} made essential use of a connection with applicative bisimilarity \cite{Abramsky90}.
	Applicative bisimilarity gives us a technique for proving equivalence between higher-order functional programs.
	In \cite{Relational}, applicative bisimilarity is defined for effectful programs using relators.
	We will briefly look at how we can derive such a relator from algebraic relations and EM-algebras.
	
	\begin{defi}[\cite{Levy11,Relational}]
		A \emph{relator} $\Gamma$ for a monad $M$ is a family of operations, giving for each pair of sets $X$, $Y$, a function $\Gamma_{X,Y}$ sending relations $\mathcal{R} \subseteq X \times Y$ to relations $\Gamma(\mathcal{R}) \subseteq MX \times MY$, such that:
		\begin{align*}
			& (1) \ =_{MX} \subseteq \Gamma(=_X),  \qquad (2) \ \Gamma(\mathcal{R})\Gamma(\mathcal{S}) \subseteq \Gamma(\mathcal{RS}), \qquad (3) \ \mathcal{R} \subseteq \mathcal{S} \Rightarrow \Gamma(\mathcal{R}) \subseteq \Gamma(\mathcal{S}),\\
			& (4) \ \forall f: X \to Z, g: Y \to W, \mathcal{R} \subseteq Z \times W. \ \Gamma(\{(x,y) \mid f(x) \ \mathcal{R} \ g(y)\}) = \{(a,b) \mid M(f)(a) \Gamma(\mathcal{R}) M(g)(b)\}
		\end{align*}
	\end{defi}

	\noindent
	For instance, $T\langle - \rangle$ defined at the end of Subsection \ref{sec:fromto} is a relator. We give two ways of constructing a relator. 
	
	Firstly, given a complete algebraic relation $\IEQ$, we define the operation $\Gamma^{\IEQ}$ for the monad $\Trees{(-)}$ as follows: For any two sets $X$ and $Y$, relation $\mathcal{R} \subseteq X \times Y$, and elements $a \in \Trees{X}$ and $b \in \Trees{Y}$, $a \ \Gamma_{X,Y}^{\IEQ}(\mathcal{R}) \ b$ holds if and only if for any two functions $f: X \to \Trees{\Zero}, g: Y \to \Trees{\Zero}$: 
	$$ (\forall x \in X, y \in Y, x \ \mathcal{R} \ y \Rightarrow \Ias{f(x)}{g(y)}) \quad \implies \quad \Ias{f^*(a)}{g^*(b)} \enspace .$$
	
	Secondly, given an Eilenberg-Moore algebra $\alpha: T\Arb \to \Arb$,
	we define the operation $\Gamma^{\alpha}$ for the monad $\Trees{(-)}$ as follows: For any two sets $X$ and $Y$, relation $\mathcal{R} \subseteq X \times Y$, and elements $a \in \Trees{X}$ and $b \in \Trees{Y}$, $a \ \Gamma_{X,Y}^{\alpha}(\mathcal{R}) \ b$ holds if and only if for any two functions $f: X \to \Arb, g: Y \to \Arb$:
	$$(\forall x \in X, y \in Y, x \ \mathcal{R} \ y \Rightarrow f(x) \fimp g(y)) \quad \implies \quad \alpha(\Trees{f}(a)) \fimp \alpha(\Trees{g}(b)) \enspace .$$
	
	$\Gamma^{\alpha} = \Gamma^{\IEQ}$ if $\IEQ$ is complete. Moreover, $\IEQ$ is base-valued precisely when $\Gamma^{\IEQ}(\textit{id}_{\Nat}) = \IEQ$.
	\begin{lem}\label{lem:relator}
		If $\Arb$ is a complete lattice, and $\alpha: T\Arb \to \Arb$ is a monotone EM-algebra, then $\Gamma^{\alpha}$ is a relator.
	\end{lem}
	This lemma holds for all of the given examples.
	If moreover $\alpha$ is a morphism in the category of $\omega$-cpos, then $\Gamma^{\alpha}$ satisfies the additional properties required in \cite{Relational} in order to use Howe's method and prove that applicative bisimilarity is compatible.
	As such, $\Arb$ is required to be a complete lattice by the theory developed in \cite{Quantitative}. This is the case in all our examples.
	
	\subsection{Involutions}
	In \cite{Quantitative}, a quantitative logic is defined with the intention to specify a behavioural equivalence. 
	One optional ingredient in that definition is the notion of negation, an involution on the carrier set $\Arb$ of the Eilenberg-Moore algebra.
	Given how trees are formulated in this paper, with the addition of a top element $\top$, there is a natural candidate for an involution function on $\Trees{\Zero}$.
	
	A function $f: X \to X$ on a preorder $X$ is an involution if, a) for all $x \in X$, $f(f(x)) = x$, and b) for all $x,y \in X$,  $x \leq y$ holds if and only if $f(y) \leq f(x)$.
	Note in particular that $\Nat$ with the discrete ordering has the identity function as an involution, and $\Zero$ has a trivial involution.
	Given a preorder $X$ with involution $f$, we let $\neg(-): \Trees{X} \to \Trees{X}$ be the function that takes a tree $t \in \Trees{\Zero}$, and produces a tree of the same shape by doing the following two alterations:
	\begin{itemize} 
		\item Replace each leaf of $t$ labelled $\top$ with a leaf labelled $\bot$, and vice versa.
		\item Replace each leaf of $t$ labelled $x \in X$, with a leaf labelled $f(x)$.
	\end{itemize}
	Note that $\neg(-)$ gives an involution on $\Trees{X}$ with respect to the tree ordering $\leq_{\Trees{X}}$.
	In particular, $\Trees{\Zero}$ and $\Trees{\Nat}$ have an involution.
	
	In order for $\neg$ to induce an involution on $\clos{\IU}$, we need $\IEQ$ to preserve involutions:
	\begin{defi}
		$\IEQ \subseteq (\Trees{\Nat})^2$ \emph{preserves involutions} if $\forall a,b \in \Trees{\Nat}. \ \Ias{a}{b} \Leftrightarrow \Ias{\neg b}{\neg a}$.
	\end{defi}
	\noindent
	We call the inequation $\Ias{\neg b}{\neg a}$ the \emph{involution-complement} of $\Ias{a}{b}$.
	If $\IEQ$ preserves involutions, then the function $f: \clos{\IU} \to \clos{\IU}$ given by $\clos{t} \mapsto \clos{\neg t}$ is well-defined and hence gives an involution on $\clos{\IU}$.
	Unfortunately though, for all but the input example, $\neg$ does not give an involution on $\clos{\IU}$.
	
	For the cost example, the $\neg$ does not give an involution because $\Nat_{\infty}$ simply does not have an involution. Remember the asserted axiom $(\ETi(x) \leq x) \in \Axioms \subseteq \IEQ$, and note that its complement $(x \leq \ETi(x))$ is not contained in $\IEQ$.  
	Hence, in order to get a proper involution, we need the $\IEQ$ to preserve involutions (at least on finite trees). This is however impossible for this example.
	
	For nondeterminism too, $\neg$ does not give an involution on $\clos{\IU}$. This is curious, since its axioms are closed under involution-complement. Moreover, $\clos{\IU} = \{\bot,\Diamond,\top\}$ does have an obvious candidate for an involution.
	The problem is our treatment of infinite trees.
	Consider the tree $t = \mu x.\ENor(x,x)$, the infinite binary tree without any leaves. By admissibility, $t$ is equal to $\bot$. However, $\neg t = t$ and $\neg \bot = \top$, hence $\IEQ$ does not preserve involutions.
	
	The same problem occurs for examples like probability and global store. However, as in the case of nondeterminism, $\neg$ does give a natural notion of involution on finite trees:
	\begin{itemize}
		\item For probability, we get $\neg: [0,1] \to [0,1]$ sending $p$ to $1-p$.
		\item For global store, we get $\neg: \mathcal{P}(\Nat) \to \mathcal{P}(\Nat)$ sending $S$ to its complement $\Nat-S$.
	\end{itemize} 
	It may be possible to extend this definition to infinite trees, circumventing the above mentioned issue, using approximations to change how $\neg$ operates on infinite trees. This is a potential subject for future research. 

	\section{Extracting Boolean predicates}\label{Sec:Bool}
	In \cite{modal_journal}, effectful behaviour is described using Boolean modalities. In this section, we will see how we can extract a collection of Boolean modalities from an EM-algebra, and study when this collection correctly characterises that EM-algebra and its induced behavioural equivalence.
	
	Here we define a \emph{Boolean modality} $o$ on an effect signature $\Sigma$ in a slightly different way from \cite{modal_journal}, in order for it to fit more naturally in the current framework. 
	We specify $o$ using a map $\denote{o} : \Trees{\Bool} \to \Bool$, where $\Bool$ are the \emph{Booleans}: the two element set $\{\False,\True\}$ with an order `$\Rrightarrow$' defined by  $\False \Rrightarrow \True$. 
	For any predicate $P : X \to \Bool$ and modality $o$ we define a predicate $o(P): \Trees{X} \to \Bool$ denoted by $\denote{o} \circ \Trees{P}$.
	Given a set of modalities $\Obser$ we define a relation $\basrel \ \subseteq \Trees{\Nat}$ as follows:
	$$l \basrel r \quad :\iff \quad \forall v: \Nat \to \Bool, \forall o \in \Obser. \ \ o(v)(l) \Rrightarrow o(v)(r) \enspace .$$
	
	\noindent
	We construct a set of Boolean modalities from an EM-algebra $\alpha$ using two ingredients:
	\begin{itemize}
	\item Suppose $X$ is a set with an $\omega$-complete preorder $\leq$. 
	A predicate $P : X \to \Bool$ is \emph{open} with respect to $\leq$ if for any ascending sequence of elements $a_0 \leq a_1 \leq a_2 \leq \dots$ such that $\forall n \in \Nat. P(a_n) = \False$, $P(\bigvee_{n \in \Nat} a_n) = \False$.
	For instance, in case that $X = [0,1]$, the open predicates are given by tests $P_r(a) = \True \iff a > r$ where $r$ is some real number.
	\item Given a preorder $X$ with minimum element $\bot \in X$, and given an element $v \in X$, we define $v' : \Bool \to X$ to be the function sending $\True$ to $v$ and $\False$ to $\bot$. We then define $\widehat{v} := \Trees{v'} : \Trees{\Bool} \to \Trees{X}$.
	\end{itemize}

	\noindent
	Consider an EM-algebra $\alpha: T\Arb \to \Arb$, where $\Arb$ is a complete lattice with order $\fimp$.
	\begin{defi}
	Given an open predicate $P : \Arb \to \Bool$ and $v \in \Arb$, we define a Boolean modality $(P, v)$ with denotation $\denote{(P,v)} = (P \circ \alpha \circ \widehat{v}) : \Trees{\Bool} \to \Bool$. We define $\Obser_{\alpha} := \{(P,v) \mid v \in \Arb, P : \Arb \to \Bool \ \text{open w.r.t.} \ \fimp\}$.
	$$\xymatrix{
		\Trees{\Bool} \ar@{.>}[r]_{\denote{(P, v)}} \ar@{->}[d]_{\widehat{v}} & \Bool \\
		\Trees{\Arb} \ar@{->}[r]_{\alpha} & \Arb \ar@{->}[u]_{P}
	}$$
	\end{defi}

	Note that the EM-algebra $\alpha$ defined in~\ref{sec:fromto} factors through the function $\Trees{c}: \Trees{\Arb} \to \Trees{\Trees{\Zero}}$, choosing for each element $a \in \Arb$ an appropriate element $c(a) \in \Trees{\Zero}$.
	In this case, we see that the function $\Trees{c} \circ \widehat{v}$ assigns to $\True$ a \emph{continuation} $c(v)$. 
	Secondly, suppose $\IEQ$ is complete (hence admissible). 
	Then $\IU$ is an $\omega$-complete preorder, and the function $P \circ [-] : \Trees{\Zero} \to \Bool$ is a predicate on $\Trees{\Zero}$ which is open with respect to the order $\IU$. This motivates the following alternative definition.

	\begin{defi}
	Given a predicate $P : \Trees{\Zero} \to \Bool$ open w.r.t. $\IU$, and continuation $t \in \Trees{\Zero}$, we define a Boolean modality $(P, t)$ with denotation $\denote{(P,t)} = (P \circ \mu_{\Zero} \circ \widehat{t}) : \Trees{\Bool} \to \Bool$. We define $\Obser_{\IEQ} := \{(P,t) \mid t \in \Trees{\Zero}, P : \Trees{\Zero} \to \Bool \ \text{open w.r.t. the order} \ \IU\}$.
	$$\xymatrix{
		\Trees{\Bool} \ar@{.>}[r]_{\denote{(P, t)}} \ar@{->}[d]_{\widehat{t}} & \Bool \\
		\Trees{\Trees{\Zero}} \ar@{->}[r]_{\mu_{\Zero}} & \Trees{\Zero} \ar@{->}[u]_{P}
	}$$
	\end{defi}

	Given the discussion from before, the following lemma is evident.
	\begin{lem}\label{lem:coincidence}
		The sets of functions $\{\denote{o} \mid o \in \Obser_{\alpha_{\IEQ}}\}$ and $\{\denote{o} \mid o \in \Obser_{\IEQ}\}$ coincide.
	\end{lem}
	In particular, this means that $\basrel$ induced by $\Obser_{\IEQ}$ and $\basrel$ induced by $\Obser_{\alpha}$ are identical.
	Let us look at the order $\basrel$ given by $\Obser = \Obser_{\alpha}$.
	
	\begin{lem}\label{lem:bool_char}
		For any $l, r \in \Trees{\Nat}$:
	$$ l \basrel r \quad \iff \quad \forall v \in \Arb, f : \Nat \to \{\bigvee \emptyset, v\}. \ \alpha(\Trees{f}(l)) \fimp \alpha(\Trees{f}(r)) \enspace .$$
	\end{lem}
	\begin{proof}
		Note that for $a, b \in \Arb$, $a \fimp b$ if and only if $P(a) \Rrightarrow P(b)$ for any open predicate $P: \Arb \to \Bool$. This is because open predicates preserve order, and the predicate $P_b : \Arb \to \Bool$ defined by $P_b(c) = \False \iff (c \fimp b)$ is an open predicate. The result follows by unfolding the definition of $\basrel$.
	\end{proof}

	Given the above lemma, not all complete base-valued algebraic relations on $\Trees{\Nat}$ can be expressed as $\basrel$ for some set of Boolean modalities $\Obser$. 
	As such, not all equational theories can be properly represented using Boolean modalities. 
	More concretely, it is not always the case that $\basrel$ is equal to $\Ias{}{}$.
	In such cases, we should stick to the EM-algebra as a tool for specifying distinctions for the algebraic relation.
	Later on, we will investigate this issue with Boolean modalities for two specific examples.
	First, we look at two examples for which this is not a problem.

	\subsection{Examples}
	
	\begin{exa}[Global store]
	We will start with the example of global store from Examples \ref{exa:glob1} and \ref{exa:glob}, as it most clearly conveys the use of the two ingredients for constructing Boolean modalities.
	Firstly, each open predicate $P : \mathcal{P}(\Nat) \to \Bool$, if not constantly $\False$, is specified by a finite set $f \subset \Nat$ where $P(S) = \True \iff f \subseteq S$.
	Given a value $v \in \Arb = \mathcal{P}(\Nat)$, we write $(f \mapsto v)$ for the modality $(P,v)$, where $P$ is specified by the finite set $f$.
	The following statements hold:
	\begin{itemize}
		\item $\denote{(f \mapsto v)}(\leaf{\False}) = \False$.
		\item $\denote{(f \mapsto v)}(\leaf{\True}) = \True \quad \iff \quad f \subseteq v$.
		\item $\denote{(f \mapsto v)}(\EUp{m}(t)) = \denote{(\{m\} \mapsto v)}(t)$.
		\item $\denote{(f \mapsto v)}(\ELo(m \mapsto t_m)) = \True \quad \iff \quad \forall n \in f. \ \denote{(\{n\} \mapsto v)}(t_n) = \True$.
	\end{itemize}
	We interpret $(f \mapsto v)$ as a test on trees, using a precondition and postcondition: for any starting state $n$ satisfying assertion $f$, the computation terminates with leaf $\True$ and a final state satisfying assertion $v$.
	This naturally generalises to cases where there is a finite number of global store locations.
	\end{exa}
	
	\begin{exa}[Probability]
	In the case of probability from Examples \ref{exa:prob1} and \ref{exa:prob}, the second ingredient for constructing Boolean modalities is redundant. However, this case does illustrate the importance of using open predicates, to ensure that the modalities are \emph{Scott-open} (a requirement from \cite{modal,modal_journal}).
	As noted before, open predicates on $[0,1]$ are given by tests $P_r: [0,1] \to \Bool$ such that $P_r(a) = \True \iff a > r$, where $r$ is some real number.
	For any $v \in [0,1]$:
	\begin{itemize}
		\item $\denote{(P_r,v)}(\leaf{\False}) = \True \quad \iff \quad r < 0$.
		\item $\denote{(P_r,v)}(\leaf{\True}) = \True \quad \iff \quad r < v$.
		\item $\denote{(P_r,v)}(\EPor(l,r)) = \True \quad \iff \quad \exists a,b \in \Real. \ (a+b)/2 \geq r \ \wedge \ \denote{(P_a,v)}(l) = \True \ \wedge \ \denote{(P_b,v)}(r) = \True$.
	\end{itemize}
	Note that $(P_r,v)(t) = (P_{r/v},1)(t)$ if $v \neq 0$, and $(P_r,0)(t) = \True \ \iff \ r < 0$.
	\end{exa}
	
	\subsection{Representability}
	We will now study which algebraic relations $\IEQ$ can be specified by a set of Boolean modalities.
	
	\begin{defi}
		$\IEQ$ is \emph{single-valued} if for any $a,b \in \Trees{\Nat}$:
		$$(\forall f: \Nat \to \{\bot,\leaf{0}\} \subseteq \Trees{\Nat}. \ \ \Ias{f^*(a)}{f^*(b)}) \ \implies \ \Ias{a}{b} \enspace . $$
	\end{defi}

	\begin{lem}
		Let $\IEQ$ be a complete base-valued algebraic relation and let $\Obser := \Obser_{\IEQ}$, then $\IEQ$ is single-valued if and only if $\sqsubseteq_{\Obser} = \Ias{}{}$.
	\end{lem}
	\begin{proof}
		Let $\alpha$ be the EM-algebra $\alpha_{\IEQ}$ and $\Obser = \Obser_{\IEQ}$.
		By Lemma~\ref{lem:bool_char} we derive that:
		\begin{align*}
		a \ \sqsubseteq_{\Obser} \ b \quad \iff & \quad \forall v \in \Arb, \forall f : \Nat \to \{\bigvee \emptyset, v\}. \ \alpha(\Trees{f}(a)) \fimp \alpha(\Trees{f}(b))\\ 
		\quad \iff & \quad \forall f : \Nat \to \{\bot,\leaf{0}\}, \forall v \in \Arb, \forall g : \Nat \to \{\bigvee \emptyset,v\}. \ \alpha(\Trees{g}(f^*(a))) \fimp \alpha(\Trees{g}(f^*(b))) \\
		\quad \iff & \quad \forall f : \Nat \to \{\bot,\leaf{0}\}. \ f^*(a) \ \sqsubseteq_{\Obser} f^*(b) \ .
		\end{align*}
		So if $\sqsubseteq_{\Obser} = \Ias{}{}$, then $\IEQ$ is single-valued.
		Assume that $\IEQ$ is single-valued.
		Using that $\IEQ$ is also base-valued, we do the following derivation.
		\begin{align*}
		a \ \sqsubseteq_{\Obser} \ b & \quad \iff \quad \forall P : \Nat \to \Bool, o \in \Obser_{\IEQ}. \ o(P)(a) \Rrightarrow o(P)(b)\\
		\quad \iff & \quad \forall P : \Nat \to \Bool, \forall t \in \Trees{\Zero}, \forall K : \Trees{\Zero} \to \Bool \ \text{open w.r.t.} \ \IEQ. \ (K \circ \mu_{\Zero} \circ \widehat{t})(\Trees{P}(a)) \Rrightarrow (K \circ \mu_{\Zero} \circ \widehat{t})(\Trees{P}(b))\\
		\quad \iff & \quad \forall P : \Nat \to \Bool, \forall t \in \Trees{\Zero}. \ \Ias{(\mu_{\Zero} \circ \widehat{t})(\Trees{P}(a))}{(\mu_{\Zero} \circ \widehat{t})(\Trees{P}(b))}\\
		\quad \iff & \quad \forall f : \Nat \to \{\bot,\leaf{0}\}, \forall g: \Nat \to \Trees{\Zero}. \ \Ias{g^*(f^*(a))}{g^*(f^*(b))}\\
		\quad \iff & \quad  \forall f : \Nat \to \{\bot,\leaf{0}\}. \ \Ias{f^*(a)}{f^*(b)}\\
		\iff & \quad \Ias{a}{b} \enspace .
		\end{align*}
	\end{proof}

	The examples of global store and probability are single-valued. However, not all equational theories are single-valued. Consider for instance the combination of nondeterminism with probability. 
	
	\begin{exa}[Nondeterminism + Probability]
	We look at the combination of effects given in Examples~\ref{exa:comb1} and \ref{exa:comb}, and let $\IEQ$ be the smallest complete single-valued relation containing the axioms from that example.
	We derive that $\IEQ$ contains an undesirable algebraic equation.
	\begin{lem}
		Let $\IEQ$ be as above, then $\IEQ$ contains the algebraic equation:
		$$ \EPor(\ENor(x, y), \ENor(x, z)) \ = \ \ENor(x, \EPor(y, z)) $$
	\end{lem}
	\begin{proof}
		We use single-valuedness to prove the equation. We check all possible substitutions, leaving out two cases because of symmetry, and unifying the four cases where $y=z$ using arbitrary $a,b \in \{\bot,0\}$.
		
		\medskip
		\noindent
		\begin{tabular}{c c c | l c r}
		x & y & z & $\EPor(\ENor(x, y), \ENor(x, z))$ \hspace{-20mm} & & \hspace{-20mm} $\ENor(x, \EPor(y, z))$\\
		\hline
		$a$ & $b$ & $b$ & $\EsPor(\EsNor(a, b), \EsNor(a, b))$ & $= \EsNor(a,b) =$ & $\EsNor(a, \EsPor(b, b))$\\
		$\bot$ & $0$ & $\bot$ & $\EsPor(\EsNor(\bot, 0), \EsNor(\bot, \bot))$ \hspace{-3mm} & $= \EsPor(\bot,\EsNor(\bot,0)) = \EsNor(\EsPor(\bot,\bot),\EsPor(\bot,0)) =$  & $\EsNor(\bot, \EsPor(0, \bot))$\\
		$0$ & $0$ & $\bot$ & $\EsPor(\EsNor(0, 0), \EsNor(0, \bot))$ & $= \EsPor(0,\EsNor(0,\bot)) = \EsNor(\EsPor(0,0),\EsPor(0,\bot)) =$ & $\EsNor(0, \EsPor(0, \bot))$
		\end{tabular}
	\end{proof}
	\noindent
	Informally, the equation from above asserts that the scheduler controlling nondeterministic choice knows how future probabilistic choices will be resolved.
	Other undesirable equations can be derived from this.
	\end{exa}
	
	Another problematic example is the combination of nondeterminism and global store, which we shall not discuss here. As a last example, we look at exception catching.
	
	\begin{exa}[Exception Catching]
	We cannot represent exception catching from Example \ref{exa:exep} with Boolean modalities. 
	This is evident from the following lemma.
	\begin{lem}
		Assume $\IEQ$ is complete, single-valued, and contains the axiomatic equations from Example~\ref{exa:exep}, then $\Ias{}{} = (\Trees{\Nat})^2$.
	\end{lem}
	\begin{proof}
		We first use single-valuedness to prove that $\ECa{e}(x,y) \leq x$:
		
		\smallskip
		\begin{tabular}{c c | c c c}
			$x$ & $y$ & $\ECa{e}(x,y)$ & & $x$ \\
			\hline
			$\bot$ & $\bot$ & $\ECa{e}(\bot,\bot)$ & $=$ & $\bot$ \\
			$0$ & $\bot$ & $\ECa{e}(0,\bot)$ & $\leq \ECa{e}(0,0) =$ & $0$ \\
			$\bot$ & $0$ & $\ECa{e}(\bot,0)$ & $=$ & $\bot$ \\
			$0$ & $0$ & $\ECa{e}(0,0)$ & $=$ & $0$
		\end{tabular}
		\smallskip
	
		\noindent
		If we substitute $\EEx{e}$ for $x$, and $\top$ for $y$, we get $\top = \ECa{e}(\EEx{e},\top) \leq \EEx{e} \leq \top$ and hence $\top = \EEx{e}$. 
		So $\top = \ECa{e}(\top,x) = \ECa{e}(\EEx{e},x) = x$, and we conclude that everything is equal to $\top$.
	\end{proof}
	If there are multiple exceptions, we could have avoided the use of $\top$ in the above proof. As such, the problem remains in a setting where we do not use a top element.
	\end{exa}

	\section{Conclusions}
	Comparing the theory presented in this paper with the established literature \cite{modal_journal,Quantitative}, a main difference is the addition of a top element $\top$ in the definition of effect trees $TX$, and the accommodated focus on $\Trees{\Zero}$ as an alternative to $\Trees{\{\ast\}}$. 
	If instead, we generated our EM-algebra from the latter structure, the resulting value space $\clos{\IU}$ would not be a complete lattice for some of the usual examples.
	A complete lattice is necessary for formulating program equivalence for higher-order functions, as quantitative formulas from \cite{Quantitative} need to be closed under suprema in general, and relators are formulated more easily using suprema (see Lemma \ref{lem:relator}).
	
	For instance, for global store the value constructed from the algebraic axioms is $\clos{\IU} = \mathcal{P}(\Nat)$, whereas 
	an alternative construction on $\Trees{\{\ast\}}$ yields as value space the set of partial functions from $\Nat$ to $\Nat$, which is not a complete lattice.
	
	One of the main motivations for this work was to more closely compare the two complementing views of specifying program equivalence: determining equality, and determining distinctness (inequality).
	The developed theory will hopefully give rise to practical proof methods for showing program equivalence and non-equivalence.
	Moreover, there might be some semi-decidability results in the sense of \cite{Escardo_decid}, or even decidability results in the absence of general recursion.
	
	Although this paper only studies equivalence between effectful expressions at base types $\Nat$ and $\Zero$, it is done so with the expectation that we can extend it to a description of program equivalence on functional languages with higher-order types.
	On the one hand, the EM-algebra is used as a basis to formulate a quantitative logic of behavioural properties in \cite{Quantitative}, which are used to describe differences between inequivalent functional programs.
	On the other hand, applicative bisimilarity is used in order to prove equivalence between higher-order programs, using the relator formulated in Subsection~\ref{sub:relator} in the definition of effectful applicative bisimilarity from \cite{Relational}.
	Further up-to techniques could potentially be utilised to prove equivalence of terms even more easily \cite{Sangiorgi_book, Sangiorgi_book_2, effectful_normal}.
	
	Last but not least, it would make sense to use EM-algebras for formulating quantitative relations between programs, e.g. \emph{metrics} \cite{ARNOLD1980,Escardo_metric,MardarePP16,appdis}.
	This is a potential subject for future research.

	\bibliographystyle{eptcs}
	\bibliography{biblio}
	
\end{document}